\documentclass[11pt]{article}

\usepackage{algorithm}
\usepackage{algpseudocode}
\usepackage[normalem]{ulem}
\algtext*{EndWhile}
\algtext*{EndIf}
\algtext*{EndFor}

\newcommand{\rr}{\mathbb{R}}
\newcommand{\nn}{\mathbb{N}}

\newcommand{\ra}{\rightarrow}

\newcommand{\calT}{\mathcal{T}}
\newcommand{\calO}{\mathcal{O}}
\newcommand{\expect}[1]{\mathbb{E}\left[#1\right]}
\newcommand{\Gs}{G^*}

\usepackage[utf8]{inputenc}
\usepackage[T1]{fontenc}
\usepackage{lmodern}
\usepackage[DIV=11]{typearea} 

\usepackage{microtype}

\usepackage{amssymb}
\usepackage{amsmath}
\usepackage{amsthm}
\usepackage{thmtools}

\usepackage{tikz}

\usepackage{xcolor}
\usepackage{xspace}
\usepackage{paralist}
\usepackage{comment}

\usepackage[
	style=alphabetic,
	backref=true,
	doi=false,
	url=false,
	maxcitenames=3,
	mincitenames=3,
	maxbibnames=10,
	minbibnames=10,
	backend=bibtex8,
	sortlocale=en_US,
	sorting=nyt,
	bibencoding=ascii
]{biblatex}

\usepackage[ocgcolorlinks]{hyperref} 

\newcommand*\samethanks[1][\value{footnote}]{\footnotemark[#1]}

\makeatletter
\g@addto@macro\bfseries{\boldmath}
\makeatother

\makeatletter
\g@addto@macro\mdseries{\unboldmath}
\g@addto@macro\normalfont{\unboldmath}
\g@addto@macro\rmfamily{\unboldmath}
\g@addto@macro\upshape{\unboldmath}
\makeatother

\DeclareCiteCommand{\citem}
    {}
    {\mkbibbrackets{\bibhyperref{\usebibmacro{postnote}}}}
    {\multicitedelim}
    {}

\renewcommand*{\multicitedelim}{\addcomma\space}

\AtEveryCitekey{\clearfield{note}}

\newcommand{\myhref}[1]{%
  \iffieldundef{doi}
    {\iffieldundef{url}
       {#1}
       {\href{\strfield{url}}{#1}}}
    {\href{http://dx.doi.org/\strfield{doi}}{#1}}%
}

\DeclareFieldFormat{title}{\myhref{\mkbibemph{#1}}}
\DeclareFieldFormat
  [article,inbook,incollection,inproceedings,patent,thesis,unpublished]
  {title}{\myhref{\mkbibquote{#1\isdot}}}

\makeatletter
\AtBeginDocument{%
    \newlength{\temp@x}%
    \newlength{\temp@y}%
    \newlength{\temp@w}%
    \newlength{\temp@h}%
    \def\my@coords#1#2#3#4{%
      \setlength{\temp@x}{#1}%
      \setlength{\temp@y}{#2}%
      \setlength{\temp@w}{#3}%
      \setlength{\temp@h}{#4}%
      \adjustlengths{}%
      \my@pdfliteral{\strip@pt\temp@x\space\strip@pt\temp@y\space\strip@pt\temp@w\space\strip@pt\temp@h\space re}}%
    \ifpdf
      \typeout{In PDF mode}%
      \def\my@pdfliteral#1{\pdfliteral page{#1}}
      \def\adjustlengths{}%
    \fi
    \ifxetex
      \def\my@pdfliteral #1{}
      \def\adjustlengths{\setlength{\temp@h}{-\temp@h}\addtolength{\temp@y}{1in}\addtolength{\temp@x}{-1in}}%
    \fi%
    \def\Hy@colorlink#1{%
      \begingroup
        \ifHy@ocgcolorlinks
          \def\Hy@ocgcolor{#1}%
          \my@pdfliteral{q}%
          \my@pdfliteral{7 Tr}
        \else
          \HyColor@UseColor#1%
        \fi
    }%
    \def\Hy@endcolorlink{%
      \ifHy@ocgcolorlinks%
        \my@pdfliteral{/OC/OCPrint BDC}%
        \my@coords{0pt}{0pt}{\pdfpagewidth}{\pdfpageheight}%
        \my@pdfliteral{F}
        \my@pdfliteral{EMC/OC/OCView BDC}%
        \begingroup%
          \expandafter\HyColor@UseColor\Hy@ocgcolor%
          \my@coords{0pt}{0pt}{\pdfpagewidth}{\pdfpageheight}%
          \my@pdfliteral{F}
        \endgroup%
        \my@pdfliteral{EMC}%
        \my@pdfliteral{0 Tr}
        \my@pdfliteral{Q}%
      \fi
      \endgroup
    }%
}
\makeatother

\newif\iffullversion
\newcommand{\infull}[1]{\iffullversion #1\fi}

\fullversiontrue

\usetikzlibrary{arrows,shapes}

\colorlet{DarkRed}{red!50!black}
\colorlet{DarkGreen}{green!50!black}
\colorlet{DarkBlue}{blue!50!black}

\hypersetup{
	linkcolor = DarkRed,
	citecolor = DarkGreen,
	urlcolor = DarkBlue,
	bookmarks = true,
	bookmarksnumbered = true,
	linktocpage = true
}

\declaretheorem[numberwithin=section]{theorem}
\declaretheorem[numberlike=theorem]{lemma}
\declaretheorem[numberlike=theorem]{claim}

\declaretheorem[numberlike=theorem]{fact}
\declaretheorem[numberlike=theorem]{corollary}
\declaretheorem[numberlike=theorem]{definition}

\newenvironment{pfof}[1]{\begin{proof}[\textbf{Proof of #1: }]}{\end{proof}}

\title{Graph Minors for Preserving Terminal Distances Approximately -- Lower and Upper Bounds\thanks{An extended abstract will appear in Proceedings of ICALP 2016.}}

\author{
Yun Kuen Cheung\thanks{University of Vienna, Faculty of Computer Science, Vienna, Austria.}
\and Gramoz Goranci\samethanks[2]
\and Monika Henzinger\samethanks[2]
}

\date{}

\hypersetup{
	pdftitle = {title},
	pdfauthor = {authors}
}

\begin{document}
\maketitle

\begin{abstract}
Given a graph where vertices are partitioned into $k$ terminals and non-terminals,
the goal is to compress the graph (i.e., reduce the number of non-terminals) using minor operations while preserving terminal distances approximately.
The distortion of a compressed graph is the maximum multiplicative blow-up of distances between all pairs of terminals.
We study the trade-off between the number of non-terminals and the distortion.
This problem generalizes the Steiner Point Removal (SPR) problem, in which all non-terminals must be removed.

We introduce a novel black-box reduction to convert any lower bound on distortion for the SPR problem
into a super-linear lower bound on the number of non-terminals, with the same distortion, for our problem.
This allows us to show that there exist graphs such that every minor with distortion less than $2~/~2.5~/~3$
must have $\Omega(k^2)~/~\Omega(k^{5/4})~/~\Omega(k^{6/5})$ non-terminals, plus more trade-offs in between.
The black-box reduction has an interesting consequence: if the tight lower bound on distortion for the SPR problem is super-constant,
then allowing any $\calO(k)$ non-terminals will \emph{not} help improving the lower bound to a constant.

We also build on the existing results on spanners, distance oracles and connected 0-extensions
to show a number of upper bounds for general graphs, planar graphs, graphs that exclude a fixed minor and bounded treewidth graphs.
Among others, we show that any graph admits a minor with $\calO(\log k)$ distortion and $\calO(k^{2})$ non-terminals,
and any planar graph admits a minor with $1+\varepsilon$ distortion and $\widetilde{\mathcal{O}}((k/\varepsilon)^{2})$ non-terminals.
\end{abstract}


\section{Introduction}

\emph{Graph compression} generally describes a transformation of a \emph{large} graph $G$ into a \emph{smaller} graph $G'$ that preserves,
either exactly or approximately, certain features (e.g., distance, cut, flow) of $G$.
Its algorithmic value is apparent, since the compressed graph can be computed in a preprocessing step of an algorithm,
so as to reduce subsequent running time and memory.
Some notable examples are graph spanners, distance oracles and cut/flow sparsifiers.

In this paper, we study compression using minor operations, which has attracted increasing attention in recent years.
Minor operations include vertex/edge deletions and edge contractions.
It is naturally motivated since it preserves certain structural properties of the original graph, e.g., any minor of a planar graph remains planar, while
reducing the size of the graph.
We are interested in \emph{vertex sparsification}, where $G$ has a designated subset $T$ of $k$ vertices called the \emph{terminals},
and the goal is to reduce the number of non-terminals in $G'$ while preserving some feature among the terminals.
Recent work in this field studied preserving cuts and flows.
Our focus here is on preserving terminal distances approximately in a multiplicative sense, i.e.,
we want that for any terminals $t,t'$, $d_G(t,t') \leq d_{G'}(t,t') \leq \alpha\cdot d_G(t,t')$, for a small \emph{distortion} $\alpha$.
This problem, called {\em Approximate Terminal Distance Preservation (ATDP) problem}, has natural applications in
multicast routing~\cite{ChuRZ2000} and network traffic optimization~\cite{ScharfWZ2015}.
It was also suggested in~\cite{distancepreserving} that to solve the \emph{subset travelling salesman problem},
one can compute a compressed minor with a small distortion as a preprocessing step for algorithms
that solve the travelling salesman problem for planar graphs.

ATDP was initiated by Gupta~\cite{gupta01}, who introduced the related {\em Steiner Point Removal (SPR) problem}:
Given a tree $G$ with both terminals and non-terminals,
output a weighted tree $G'$ {\em with terminals only} which minimizes the distortion.
Gupta gave an algorithm that achieves a distortion of $8$.
Chan et al.~\cite{chan} observed that Gupta's algorithm returned always  a minor of $G$.
For general graphs, Kamma et al.~\cite{KammaKN2015} gave an algorithm to construct a minor
with distortion $\calO(\log^5 k)$.
Krauthgamer et al.~\cite{distancepreserving} studied ATDP and
showed that every graph has a minor with $\calO(k^4)$ non-terminals and distortion $1$.
It is then natural to ask, for different classes of graphs,
what the trade-off between the distortion and the number of non-terminals is.
In this paper, for different classes of graphs, and w.r.t.~different allowed distortions,
we provide lower and upper bounds on the number of non-terminals needed.



\subsection*{Further Related Work}

Basu and Gupta~\cite{BasuG2008} showed that for outer-planar graphs,
SPR can be solved with distortion $\calO(1)$.
When randomization is allowed, Englert et al.~\cite{englert10} showed that for graphs that exclude a fixed minor,
one can construct a randomized minor  for SPR with $\calO(1)$ expected distortion.
It remains open whether similar guarantees can be obtained in the deterministic setting.
%
Krauthgamer et al.~\cite{distancepreserving} showed that solving ATDP with distortion $1$ for planar graphs
needs $\Omega(k^2)$ non-terminals.

In the past few years, there has been a considerable amount of work on cut/flow vertex sparsifiers~\cite{moitra09, leighton, charikar, mm10, englert10, juliasteiner, andoni, racke2014}. In this setting, given a capacitated graph $G$ with terminals $T \subset V$, the goal is to find a sparsifier $H$ with $V(H)=T$ preserving all terminal cuts up to a factor $q \geq 1$, i.e. for all $S \subset T$, $\text{mincut}_G(S, T \setminus S) \leq \text{mincut}_H(S, T \setminus S) \leq  q \cdot \text{mincut}_G(S, T \setminus S)$. It is worth pointing out that in some setting, there is an equivalence between the construction of vertex cut/flow and distance sparsifiers~\cite{racke08, englert10}. 

A related graph compression is spanners, where the objective is to reduce the number of edges by edge deletions only.
We will use a spanner algorithm (e.g.,~\cite{AlthoferDDJS1993}) to derive our upper bound results for general graphs.
Although spanner operation enjoys much less freedom than minor operation,
proving a lower bound result for it is notably difficult.
Assuming the Erd\"{o}s girth conjecture~\cite{Erdos1963}, there are lower bounds that match the best known upper bounds,
but the conjecture seems far from being settled \cite{Wenger1991}.
Woodruff~\cite{Woodruff2006} showed a lower bound result bypassing the conjecture, but only for \emph{additive} spanners.



\subsection*{Our Contributions}

For various classes of graphs, we show lower and upper bounds on the number of non-terminals needed
in the minor for low distortion.
The table below summarizes our results. 
\begin{table}[H]
\begin{center}
\begin{tabular}{|l|c|c|}
\hline
Graph &  Upper Bound & Lower Bound   \\ \hline
      &  (distortion, size) & (distortion, size)   \\ \hline  \hline
General & $\forall q \in \mathbb{N}$ \quad  $(2q-1,\mathcal{O}(k^{2+2/q}))$ & $(2-\varepsilon, \Omega(k^2))$   \\
General & $-$ & $(2.5-\varepsilon, \Omega(k^{5/4}))$,~ $(3-\varepsilon, \Omega(k^{6/5}))$   \\
 &  & \small{(see Theorem \ref{thm:lower-bound-distortion-25} for more guarantees)}    \\
General & $-$ & $(2-\varepsilon,\varepsilon^3 k^2 /150)$-rand   \\
B.-Treewidth $p$ & $\forall q \in \mathbb{N}$ \quad $ (2q-1,\mathcal{O}(p^{1+2/q}k))$ &  $(1, \Omega(pk))$~\cite{distancepreserving}    \\ 
Exc.-Fix.-Minor &  $(\mathcal{O}(1),\widetilde{\mathcal{O}}(k^{2})$  & $-$    \\ 
Planar &  $(3,\widetilde{\mathcal{O}}(k^{2}))$,~ $(1+\varepsilon,\widetilde{\mathcal{O}}((k/\varepsilon)^{2})$ &  $(1+o(1), \Omega(k^{2}))$~\cite{distancepreserving}    \\ \hline
General & $(\mathcal{O}(\log^{5} k),0)$~\cite{KammaKN2015} & $-$ \\
Outerplanar & $(\mathcal{O}(1),0)$~\cite{BasuG2008} & $-$ \\
Trees & $(8,0)$~\cite{gupta01} & $(8-o(1),0)$~\cite{chan} \\ \hline
General & $(\mathcal{O}(\log k),0)$-rand~\cite{englert10} & $-$ \\
Exc.-Fix.-Minor & $(\mathcal{O}(1),0)$-rand~\cite{englert10} & $(2-o(1),0)$-rand
\\ \hline
\end{tabular}
\caption{The results which are \emph{not} followed by a reference are shown in this paper.
The guarantees with the extension ``-rand'' refer to \emph{randomized} distance approximating minors;
``size'' refers to the number of non-terminals in the minor.}
\end{center}
\end{table}
For our lower bound results, we use a novel black-box reduction to convert any lower bound on distortion for the SPR problem
into a super-linear lower bound on the number of non-terminals for ATDP with the same distortion.
Precisely, we show that given any graph $G^*$ such that solving its SPR problem leads to a minimum distortion of $\alpha$,
we use $G^*$ to construct a new graph $G$ such that every minor of $G$ with distortion less than $\alpha$ must have
at least $\Omega(k^{1+\delta(\Gs)})$ non-terminals, for some constant $\delta(\Gs) > 0$.
The lower bound results in the above table are obtained by using for $G^*$ a complete ternary tree of height $2$,
which was shown that solving its SPR problem leads to minimum distortion $3$~\cite{gupta01}.
More trade-offs are shown by using for $G^*$ a complete ternary tree of larger heights.

The black-box reduction has an interesting consequence.
For the SPR problem on general graphs, there is a huge gap between the best known lower and upper bounds, which are
$8$~\cite{chan} and $\calO(\log^5 k)$~\cite{KammaKN2015}; it is unclear what the asymptotically tight bound would be.
Our black-box reduction allows us to prove the following result concerning the tight bound:
for general graphs, if the tight bound on distortion for the SPR problem is super-constant,
then for any constant $C>0$, even if $Ck$ non-terminals are allowed in the minor, the lower bound will remain super-constant.
See Theorem \ref{thm:spr-vs-lspr} for a formal statement of this result.

We also build on the existing results on spanners, distance oracles and connected 0-extensions to show a number of upper bound results for general graphs, planar graphs and graphs that exclude a fixed minor. Our techniques, combined with an algorithm in Krauthgamer et al.~\cite{distancepreserving}, yield an upper bound result for graphs with bounded treewidth. In particular, our upper bound on planar graphs implies that allowing quadratic number of non-terminals, we can construct a deterministic minor with arbitrarily small distortion. 

\section{Preliminaries}

Let $G=(V,E,\ell)$ denote an undirected graph with terminal set $T\subset V$ of cardinality $k$,
where $\ell:E\ra\rr^+$ is the length function over edges $E$. 
A graph $H$ is a \emph{minor} of $G$ if $H$ can be obtained from $G$ by performing a sequence of vertex/edge deletions and edge contractions,
but no terminal can be deleted, and no two terminals can be contracted together.
In other words, all terminals in $G$ must be \emph{preserved} in $H$.

Besides the above standard description of minor operations,
there is another equivalent way to construct a minor $H$ from $G$~\cite{KammaKN2015},
which will be more convenient for presenting some of our results.
A partial partition of $V(G)$ is a collection of pairwise disjoint subsets of $V(G)$ (but their union can be a proper subset of $V(G)$).
Let $S_1,\cdots,S_m$ be a partial partition of $V(G)$ such that (1) each induced graph $G[S_i]$ is connected,
(2) each terminal belongs to exactly one of these partial partitions, and (3) no two terminals belong to the same partial partition.
Contract the vertices in each $S_i$ into one single ``super-node'' in $H$.
For any vertex $u\in V(G)$, let $S(u)$ denote the partial partition that contains $u$;
for any super-node $u\in V(H)$, let $S(u)$ denote the partial partition that is contracted into $u$.
In $H$, super-nodes $u_1,u_2$ are adjacent \emph{only if} there exists an edge in $G$ with one of its endpoints in $S(u_1)$ and the other in $S(u_2)$.
We denote the super-node that contains terminal $t$ by $t$ as well.

\begin{definition}
The graph $H=(V',E',\ell')$ is an $\alpha$-distance approximating minor (abbr.~$\alpha$\emph{-DAM}) of $G=(V,E,\ell)$
if $H$ is a minor of $G$ and for any $t,t' \in T$, $ d_G(t,t') \leq d_H(t,t') \leq \alpha \cdot d_G(t,t')$.
$H$ is an $(\alpha,y)$\emph{-DAM} of $G$ if $H$ is an $\alpha$\emph{-DAM} of $G$ with at most $y$ non-terminals.
\end{definition}

We note that the SPR problem is equivalent to finding an $(\alpha,0)$-DAM.
One can also define a randomized version of distance approximating minor:

\newcommand{\calD}{\pi}

\begin{definition}
Let $\calD$ be a probability distribution over minors of $G=(V,E,\ell)$.
We call $\calD$  an $\alpha$-randomized distance approximating minor (abbr.~$\alpha$\emph{-rDAM}) of $G$
if for any $t,t' \in T$,
$\mathbb{E}_{H\sim\calD}\left[d_H(t,t')\right] \leq \alpha \cdot d_G(t,t')$,
and for every minor $H$ in the support of $\calD$, $d_H(t,t') \geq d_G(t,t')$.
Furthermore, we call $\calD$ an $(\alpha,y)$\emph{-rDAM} if $\calD$ is an $\alpha$\emph{-rDAM} of $G$,
and every minor in the support of $\calD$ has at most $y$ non-terminals.
\end{definition}

%

\section{Deterministic and Randomized Lower Bounds}\label{sect:lower-bound}

For all the lower bound results, we use a tool in combinatorial design called \emph{Steiner system}
(or alternatively, \emph{balanced incomplete block design}). Let $[k]$ denote the set $\{1,2,\cdots,k\}$.

\begin{definition}
Given a ground set $T = [k]$, an $(s,2)$-Steiner system (abbr.~$(s,2)$-SS) of $T$ is a collection of $s$-subsets of $T$,
denoted by $\calT = \left\{T_1,\cdots,T_r\right\}$,
where $r = \binom{k}{2}\left/\binom{s}{2}\right.$,
such that every $2$-subset of $T$ is contained in \emph{exactly} one of the $s$-subsets.
\end{definition}

\begin{lemma}[\cite{Wilson1975}]\label{lem:SSS-exist}
For any integer $s\geq 2$, there exists an integer $M_s$ such that for every $q\in\nn$,
the set $[M_s + qs(s-1)]$ admits an $(s,2)$-\emph{SS}.
\end{lemma}

Our general strategy is to use the following black-box reduction, which proceeds by taking a \emph{small} connected graph $\Gs$ as input,
and it outputs a \emph{large} graph $G$ which contains many disjoint embeddings of $\Gs$.
Here is how it exactly proceeds:
\begin{itemize}
\item Let $\Gs$ be a graph with $s\geq 2$ terminals and $q\geq 1$ non-terminals.
Let $k$ be an integer, as given in Lemma \ref{lem:SSS-exist}, such that the terminal set $T = [k]$ admits an $(s,2)$-SS $\calT$.
\item We construct $\calT ' \subseteq \calT$ that satisfies \emph{certain} property depending on the specific problem.
For each $s$-set in $\calT'$, we add $q$ non-terminals to the $s$-set, which altogether form a \emph{group}.
The union of vertices in all groups is the vertex set of our graph $G$.
We note that each terminal may appear in many groups, but each non-terminal appears in one group only.
\item \emph{Within} each of the groups, we embed $\Gs$ in the natural way.
\end{itemize}

The following two lemmas describe some basic properties of all minors of $G$ output by the black-box above.
Their proofs are deferred to Appendix \ref{app:basic-of-minor}.
%
%
%

\begin{lemma}\label{lem:unique-group-for-edge}
Let $H$ be a minor of $G$. Then for each edge $(u_1,u_2)$ in $H$, there exists exactly one group $R$ in $G$
such that $S(u_1)\cap R$ and $S(u_2)\cap R$ are both non-empty.
\end{lemma}

The above lemma permits us to legitimately define the notion $R$-edge:
an edge $(u_1,u_2)$ in $H$ is an $R$-edge if $R$ is the unique group that intersects both $S(u_1)$ and $S(u_2)$.

\begin{lemma}\label{lem:interchange-at-terminal}
Suppose that in a minor $H$ of $G$,
$(u_1,u_2)$ is a $R_1$-edge and $(u_2,u_3)$ is $R_2$-edge, where $R_1\neq R_2$.
Then $R_1$ and $R_2$ intersect, and $S(u_2)$ contains the terminal in $R_1\cap R_2$.
\end{lemma}

We will show that for any minor $H$ with low distortion, at least one of the non-terminals in each group must be retained,
and thus $H$ must have at least $|\calT '|$ non-terminals.
We first present some of our main theorems on lower bounds and then prove them;
two more theorems are given in Section \ref{subsect:full-general}.

\begin{theorem} \label{thm:lower-bound-star}
For infinitely many $k\in \nn$, there exists a bipartite graph with $k$ terminals
which does not have a $(2-\epsilon,k^2/7)$-\emph{DAM}, for all $\epsilon > 0$.
\end{theorem}

\begin{theorem}\label{thm:lower-bound-distortion-25}
There exists a constant $c_1 > 0$, such that for infinitely many $k\in \nn$,
there exists a quasi-bipartite graph with $k$ terminals
which does not have an $(\alpha-\epsilon,c_1 k^\gamma)$-\emph{DAM}, for all $\epsilon>0$,
where $\alpha,\gamma$ are given in the table below.

\begin{center}
\begin{tabular}{|c||c|c|c|c|c|c|c|}
\hline
$\alpha$ & $2.5$ & $3$ & $10/3$ & $11/3$ & $4$ & $4.2$ & $4.4$ \\
\hline
$\gamma$ & $5/4$ & $6/5$ & $10/9$ & $11/10$ & $12/11$ & $21/20$ & $22/21$\\
\hline
\end{tabular}
\end{center}
\end{theorem}

\begin{theorem} \label{thm:rand-lower-bound-star}
For infinitely many $k\in \nn$, there exists a bipartite graph with $k$ terminals
which does not have a $\left(2-\epsilon, \epsilon^3 k^2 / 150\right)$-\emph{rDAM}, for any $1\geq \epsilon>0$.
\end{theorem}


\subsection{Proof of Theorem \ref{thm:lower-bound-star}}

\begin{figure}[htp]
\centering
\begin{minipage}{.35\textwidth}
  \centering
\begin{tikzpicture}
\tikzstyle{vertex}=[circle, fill=white, draw = black, minimum size = 6pt, inner sep=2pt]
\tikzstyle{vertex1}=[fill = white, draw = white]

\tikzstyle{edge}=[-,thick ]
\tikzstyle{elipse}=[-,thick ]
	
  \node[vertex] (n1) at (0,0) {$1$} ;
  \node[vertex] (n2) at (3,0)  {$6$} ;
  \node[vertex] (n3) at (1.5,3)  {$2$} ;
  \draw[edge] (n1) -- (n2) ;
  \draw[edge] (n2) -- (n3) ;
  \draw[edge] (n3) -- (n1) ;
  \draw[elipse] (1.5,0.92) ellipse (0.91 and 0.91) ;
  \node[vertex] (n4) at (1.5,0) {$5$};
  \node[vertex] (n5) at (0.7,1.35) {$4$};
  \node[vertex] (n6) at (2.3,1.35) {$7$};

  \draw[edge] (n3) -- (n4) ;
  \draw[edge] (n2) -- (n5) ;
  \draw[edge] (n1) -- (n6) ;
  \node[vertex] (n7) at (1.5,0.9)  {$3$} ;

\end{tikzpicture}
\end{minipage}%
\begin{minipage}{.65\textwidth}
  \centering
\begin{tikzpicture}
\tikzstyle{vertex}=[circle,draw = white, fill=black, minimum size = 8pt, inner sep=2pt]
\tikzstyle{vertex1}=[fill = white, draw = white]

\tikzstyle{edge}=[-,thick]
\tikzstyle{edge1}=[-,thick, gray]
\tikzstyle{vertex2}=[thick, draw = black, fill = gray!30, line width = 0.3mm] 
 
  \node[vertex] (n1) at (0,0) {} ;
  \node[vertex] (n2) at (1.25,0) {} ;
  \node[vertex] (n3) at (2.5,0) {} ;
  \node[vertex] (n4) at (3.75,0) {} ;
  \node[vertex] (n5) at (5,0) {};
  \node[vertex] (n6) at (6.25,0) {};
  \node[vertex] (n7) at (7.5,0) {} ;
  
  \node[vertex2] (n124) at (0,2) {} ;
  \node[vertex2] (n235) at (1.25,2) {} ;
  \node[vertex2] (n346) at (2.5,2) {} ;
  \node[vertex2] (n457) at (3.75,2) {} ;
  \node[vertex2] (n561) at (5,2) {} ;
  \node[vertex2] (n672) at (6.25,2) {} ;
  \node[vertex2] (n713) at (7.5,2) {} ;
  
  \draw[edge] (n124) -- (n1) node[below = 2.5pt] {$1$};
  \draw[edge] (n124) -- (n2) node[below = 2.5pt] {$2$};
  \draw[edge] (n124) -- (n4) node[below = 2.5pt] {$4$};
  
  \draw[edge] (n235) -- (n2) ;
  \draw[edge] (n235) -- (n3) node[below = 2.5pt] {$3$};
  \draw[edge] (n235) -- (n5) node[below = 2.5pt] {$5$};
  
  \draw[edge] (n346) -- (n3) ;
  \draw[edge] (n346) -- (n4) ;
  \draw[edge] (n346) -- (n6) node[below = 2.5pt] {$6$};
  
  \draw[edge] (n457) -- (n4) ;
  \draw[edge] (n457) -- (n5) ;
  \draw[edge] (n457) -- (n7) node[below = 2.5pt] {$7$};
  
  \draw[edge] (n561) -- (n5) ;
  \draw[edge] (n561) -- (n6) ;
  \draw[edge] (n561) -- (n1) ;
  
  \draw[edge] (n672) -- (n6) ;
  \draw[edge] (n672) -- (n7) ;
  \draw[edge] (n672) -- (n2) ;
  
  \draw[edge] (n713) -- (n7) ;
  \draw[edge] (n713) -- (n1) ;
  \draw[edge] (n713) -- (n3) ;
\end{tikzpicture}
\end{minipage}
\caption{On the left side: a Fano plane corresponding to a $(3,2)$-SS with $k=7$. On the right side: the bipartite graph of the Fano plane constructed using our black-box reduction. Numbered vertices are \emph{terminals} while square-shaped vertices are \emph{non-terminals}.}
\end{figure}
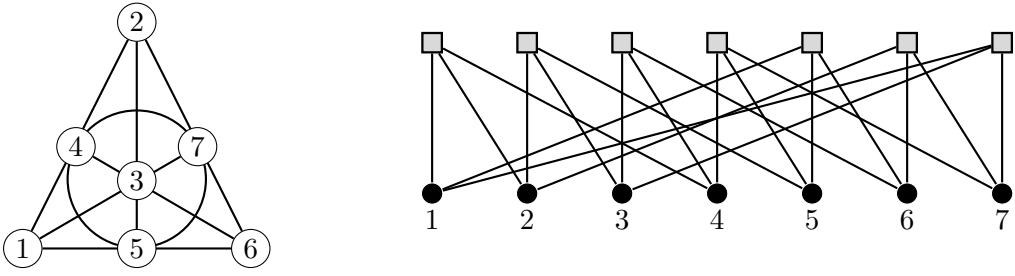

We start by reviewing the lower bound for SPR problem on stars due to Gupta~\cite{gupta01}.
 
\begin{lemma}\label{lm: gupta1}
Let $\Gs = (T \cup \{v\}, E)$ be an unweighted star with $k \geq 3$ terminals, in which $v$ is the center of the star.
Then, every edge-weighted graph only on the terminals $T$ with fewer than $\binom{k}{2}$ edges has distortion at least $2$.
\end{lemma}

We construct $G$ using the black-box reduction above. Let $k\in\nn$ be such that the terminals $T=[k]$ admits a $(3,2)$-SS, denoted by $\calT$ (see the figure above).
Here, we set $\calT' = \calT$ and $\Gs$ to be the star with $3$ terminals, as described in Lemma \ref{lm: gupta1}.

By the definition of Steiner system, the shortest path between every pair of terminal $t,t'$ in $G$ is unique,
which is the $2$-hop path within the group that contains both terminals, i.e., $d_G(t,t') = 2$ for all $t,t' \in T$.
Every other simple path between $t,t'$ must pass through an extra terminal, so the length of such simple path is at least $4$.



Let $H$ be a minor of $G$. Suppose that the number of non-terminals in $H$ is less than $r$, then there exists a group $R$ in which
its non-terminal is not retained (which means that it is either deleted, or contracted into a terminal in that group).
By Lemma \ref{lm: gupta1}, there exists a pair of terminals in that group
such that every simple path within $R$ (which means a path comprising of $R$-edges only) between the two terminals has length at least $4$.
And every other simple path must pass through an extra terminal (just as in $G$), so again it has length at least $4$.
Thus, the distortion of the two terminals is at least $2$.

Therefore, every $(2-\epsilon)$-DAM of $G$ must have $r > k^{2}/7$ non-terminals.

\subsection{Proof of Theorem \ref{thm:lower-bound-distortion-25}}

We will give the proof for the case $\alpha = 2.5$ here, and discuss how to generalize this proof for other distortions.
We will first define the notions of \emph{detouring graph} and \emph{detouring cycle},
and then use them to construct the graph $G$ that allows us to show the lower bound.

\smallskip

\noindent\textbf{Detouring Graph and Detouring Cycle.~}
For any $s\geq 3$, let $k\in\nn$ be such that the terminal set $T = [k]$ admits an $(s,2)$-SS.
Let $\calT = \{T_1,\cdots,T_r\}$ be such an $(s,2)$-SS.
A \emph{detouring graph} has the vertex set $\calT$.
By the definition of Steiner system, $\left|T_i\cap T_j\right|$ is either zero or one.
In the detouring graph, $T_i$ is adjacent to $T_j$ if and only if $\left|T_i\cap T_j\right| = 1$.
Thus, in the detouring graph, it is legitimate to give each edge $(T_i,T_j)$ a \emph{terminal label}, which is the terminal in $T_i\cap T_j$.
A \emph{detouring cycle} is a cycle in the detouring graph such that no two neighboring edges of the cycle have the same terminal label.

\smallskip

\noindent\textbf{Fact.~}Suppose that two edges in the detouring graph have a common vertex,
and their terminal labels are different, denoted by $t,t'$.
Then the common vertex must be an $s$-set in $\calT$ containing both $t,t'$.
By the definition of Steiner system, the $s$-set is uniquely determined.


\begin{claim}\label{claim:no-dcycle}
In the detouring graph, number of detouring cycles of size $\ell\geq 3$ is at most $k^\ell$.
\end{claim}

\begin{proof}
Let $(t_1,\cdots,t_\ell)$ be an $\ell$-tuple, where each entry is a terminal, that represents the terminal labels of a detouring cycle.
By the Fact above, the $\ell$-tuple determines uniquely all the vertices in the detouring cycle.
By trivial counting, the number of possible $\ell$-tuples is at most $k^\ell$,
and hence also the number of detouring cycles of size $\ell$.
\end{proof}


Our key lemma is: for any $L\geq 3$, we can retain $\Omega_s(k^{L/(L-1)})$ vertices in the detouring graph,
such that the induced graph on these vertices has \emph{no} detouring cycle of size $L$ or less.

\begin{lemma}\label{lem:large-detouring-graph}
For any integer $L\geq 3$, given a detouring graph with vertex set $\calT = \{T_1,T_2,\cdots,T_r\}$,
there exists a subset $\calT '\subset \calT$ of cardinality $\Omega_{s}(k^{L/(L-1)})$
such that the induced graph on $\calT '$ has no detouring cycle of size $L$ or less.
\end{lemma}

\begin{proof}
We choose the subset $\calT '$ by the following randomized algorithm:
\begin{enumerate}
\item Each vertex is picked into $\calT '$ with probability $\delta k^{-(L-2)/(L-1)}$, where $\delta = \delta(s) < 1$ is a positive constant
which we will derive explicitly later.
\item While (there is a detouring cycle of size $L$ or less in the induced graph of $\calT '$)\\
\hspace*{0.3in}Remove a vertex in the detouring cycle from $\calT '$
\end{enumerate}

After Step 1, $\expect{|\calT '|} = r \cdot \delta k^{-(L-2)/(L-1)} \ge \frac{\delta}{2s(s-1)} k^{L/(L-1)}$.
Using Claim \ref{claim:no-dcycle}, the expected number of detouring cycles of size $L$ or less is at most
$$\sum_{\ell=3}^L k^\ell \cdot (\delta k^{-(L-2)/(L-1)})^\ell \leq 2 \delta^3 k^{L/(L-1)}.$$
Thus, the expected number of vertices removed in Step 2 is at most $2 \delta^3 k^{L/(L-1)}$.
Now, choose $\delta = 1/\sqrt{8s(s-1)}$. By the end of the algorithm,
$$\expect{|\calT '|} \geq \frac{\delta}{2s(s-1)} k^{L/(L-1)} - 2 \delta^3 k^{L/(L-1)} = \Omega(k^{L/(L-1)}).\vspace*{-0.2in}$$
\end{proof}

\noindent\textbf{Construction of $G$ and the Proof.~}
Recall the black-box reduction.
Let $k$ be an integer such that $T=[k]$ admits a $(9,2)$-SS $\calT$.
By Lemma \ref{lem:large-detouring-graph}, we choose $\calT '$ to be a subset of $\calT$ with $|\calT'| = \Omega(k^{5/4})$,
such that the induced graph on $\calT '$ has no detouring cycle of size $5$ or less.
We choose $\Gs$ to be a complete ternary tree of height $2$, in which the $9$ leaves are the terminals.
For each $T_i\in \calT '$, we add four non-terminals to $T_i$, altogether forming a \emph{group}.

The following lemma is a direct consequence that the induced graph on $\calT'$ has no detouring cycle of size $5$ or less.

\begin{lemma}\label{lem:types-of-simple-paths}
For any two terminals $t,t'$ in the same group, let $R$ denote the group.
Then, in any minor $H$ of $G$, every simple path from $t$ to $t'$ either comprises of $R$-edges only,
or it comprises of edges from at least $5$ groups other than $R$.
\end{lemma}

\begin{pfof}{Theorem \ref{thm:lower-bound-distortion-25}}
Let $H$ be a $(2.5-\epsilon)$-DAM of $G$, for some $\epsilon>0$.
Suppose that there exists a group such that all its non-terminals are not retained in $H$.
By \cite{gupta01}, there exists a pair of terminals $t,t'$ in that group such that
every simple path between $t$ and $t'$, which comprises of edges of that group only, has length at least $3\cdot d_G(t,t')$.

By Lemma \ref{lem:types-of-simple-paths} and Lemma \ref{lem:interchange-at-terminal},
any other simple path $P$ between $t$ and $u$ passes through at least $4$ other terminals,
say they are $t_a,t_b,t_c,t_d$ in the order of the direction from $t$ to $t'$.
We denote this path by $P := t\ra t_a\ra t_b\ra t_c\ra t_d\ra t'$, by ignoring the non-terminals along the path.
Between every pair of consecutive terminals in $P$, the length is at least $2$.
Thus, the length of $P$ is at least $10$.
Since $d_G(t,t')\leq 4$, the length of $P$ is at least $2.5 \cdot d_G(t,t')$.

Thus, the length of \emph{every} simple path from $t$ to $t'$ in $H$ is at least $2.5\cdot d_G(t,t')$, a contradiction.
Therefore, at least one non-terminal in each group is retained in $H$. As there are $\Omega(k^{5/4})$ groups, we are done.
\end{pfof}

For the other results in Theorem \ref{thm:lower-bound-distortion-25},
we follow the above proof almost exactly, with the following modifications.
Set $s=3^h$ for some $h\geq 2$, and set $\Gs$ to be a complete ternary tree with height $h$, in which the leaves are the terminals.
Let $\alpha_h$ be a lower bound on the distortion for the SPR problem on $\Gs$.
Apply Lemma \ref{lem:large-detouring-graph} with some integer $h < L \leq \lceil\alpha_h h\rceil$.\footnote{Any choice of $L$
larger than $\lceil\alpha_h h\rceil$ will not improve the result.}
Following the above proof, attaining a distortion of $\min\left\{\frac{L}{h},\alpha_h\right\}-\epsilon$ needs $\Omega(k^{L/(L-1)})$ non-terminals.

The last puzzle we need is the values of $\alpha_h$.
Chan et al.~\cite{chan} proved that for complete binary trees of height $h$, $\lim_{h\ra +\infty} \alpha_h = 8$,
but they did not give explicit values of $\alpha_h$.
We apply their ideas to complete ternary tree of height $h$, to obtain explicit values for $h\leq 5$,
which are used to prove all the results in Theorem \ref{thm:lower-bound-distortion-25}.
The explicit values are $\alpha_2 = 3$, $\alpha_3 = \alpha_4 = 4$ and $\alpha_5 = 4.4$.
We discuss the details for computing these values in Appendix \ref{app:explicit-tree-value}.

%

\subsection{Full Generalization of Theorem \ref{thm:lower-bound-distortion-25}, and its Interesting Consequence}\label{subsect:full-general}

Indeed, we can set $\Gs$ as \emph{any} graph.
In our above proofs we used a tree for $\Gs$  because the only known lower bounds on distortion for the SPR problem are for trees.
If one can find a graph $\Gs$ (either by a mathematical proof, or by computer searches)
such that its distortion for the SPR problem is at least $\alpha$,
applying the black-box reduction with this $\Gs$, and  reusing the above proof show that
there exists a graph $G$ with $k$ terminals such that attaining a distortion of $\alpha-\epsilon$
needs $\Omega(k^{1+\delta(\Gs)})$ non-terminals, for some $\delta(\Gs) > 0$.

\begin{theorem}\label{thm:lower-bound-superlinear}
Let $\Gs$ be a graph with $s$ terminals, and the distance between any two terminals is between $1$ and $\beta$.
Suppose the distortion for the \emph{SPR} problem on $\Gs$ is at least $\alpha$.
Then, for any positive integer $\max\{2,\lceil\beta\rceil\} \leq L\leq \left\lceil \alpha\beta \right\rceil$,
there exists a constant $c_4 := c_4(s) > 0$, such that for infinitely many $k\in \nn$,
there exists a graph with $k$ terminals which does not have a
$\left(\min\left\{L/\beta,\alpha\right\}-\epsilon,c_4 k^{L/(L-1)}\right)$-\emph{DAM}, for all $\epsilon>0$.
\end{theorem}

The above theorem has an interesting consequence. For the SPR problem on general graphs, the best known lower bound is $8$,
while the best known upper bound is $\calO(\log^5 k)$ \cite{KammaKN2015}.
There is a huge gap between the two bounds, and it is not clear where the tight bound locates in between.
Suppose that the tight lower bound on SPR is super-constant.
Then for any positive constant $\alpha$, there exists a graph $\Gs_\alpha$ with $s(\alpha)$ terminals and some non-terminals, such that the distortion is larger than $\alpha$.
By Theorem \ref{thm:lower-bound-superlinear}, $\Gs_\alpha$ can be used to construct a family of graphs with $k$ terminals,
such that to attain distortion $\alpha$, the number of non-terminals needed is super-linear in $k$.
Recall that in SPR, no non-terminal can be retained. In other words, Theorem \ref{thm:lower-bound-superlinear} implies that:
\emph{if retaining no non-terminal will lead to a super-constant lower bound on distortion,
then having the power of retaining any linear number of non-terminals will not improve the lower bound to a constant}.

\newcommand{\lspr}{\text{\textsf{LSPR}}}

Formally, we define the following generalization of SPR problem.
Let $\lspr_y$ denote the problem that for an input graph with $k$ terminals,
find a DAM with at most $yk$ non-terminals so as to minimize the distortion;
the SPR problem is equivalent to $\lspr_0$.
\begin{theorem}\label{thm:spr-vs-lspr}
For general graphs, \emph{SPR} has super-constant lower bound on distortion
if and only if for any constant $y\geq 0$, \emph{$\lspr_y$} has super-constant lower bound on distortion.
\end{theorem}

\subsection{Proof of Theorem \ref{thm:rand-lower-bound-star}}

In this subsection we give a lower bound for rDAM. The strategy we follow will be very similar to that of Theorem \ref{thm:lower-bound-star}. In fact, one can view it as a randomized version of that proof. 
We start with the following lemma, which generalizes the deterministic SPR lower bound of Gupta in Lemma \ref{lm: gupta1} to randomized minors.

\begin{lemma} \label{lm: randomGupta}
Let $\Gs = (T \cup \{v\}, E)$ be an unweighted star with $k \geq 3$ terminals, in which $v$ is the center of the star.
Then, for every probability distribution over minors of $\Gs$ with vertex set $T$,
there exists a terminal pair with distortion at least $2(1-1/k)$.
\end{lemma}

We now continue with the construction of our input graph. For some constant $s \geq 3$ and some integer $k$, we construct a $(s,2)$-SS of the terminal set $T$
and denote it by  $\mathcal{T} = \{T_1,\ldots,T_r\}$, where $r = \binom{k}{2}/\binom{s}{2} \geq 2\binom{k}{2}/s^2$.
Similarly to the proof of Theorem \ref{thm:lower-bound-star}, we apply the black-box reduction with $\calT' = \calT$,
and set $\Gs$ as a star with $c_1$ terminals, to generate a bipartite graph $G$.
For any constant $c_1 > 0$, we define the family of minors
$$ \mathcal{L} := \{H : H \text{ is a minor of } G \text{ and } |V(H)| < \tbinom{k}{2}/c_1\}.$$

\begin{claim} \label{claim:contracted-whp}
Let $\pi$ be any probability distribution over $\mathcal{L}$.
There exists a non-terminal of $G$ that is involved in an edge contraction with probability at least $1 - s^{2}/2c_1$ under $\pi$.
\end{claim}
\begin{proof}
Suppose that for the sake of contradiction that the claim is not true. Then every non-terminal of $G$ is contracted with probability strictly less than $1 - s^{2}/2c_1$. This implies that every non-terminal of $G$ is \textit{not} contracted with probability at least $s^{2}/2c_1$, and hence
$$
	\mathbb{E}_{\pi}[  \text{number of non-terminals} ] > \frac{s^2}{2c_1} \cdot \frac{2\binom{k}{2}}{s^{2}} = \frac{\binom{k}{2}}{c_1}.
$$ 
The inequality along with the probabilistic method imply that there exists a minor $H$ in the support of $\pi$ with at least $\binom{k}{2}/c_1$ non-terminals, thus violating the properties of the members of $\mathcal{L}$, which leads to a contradiction. 
\end{proof}

\begin{pfof}{Theorem \ref{thm:rand-lower-bound-star}}
Let $v$ be the non-terminal from Claim \ref{claim:contracted-whp} and let $T_i$ be its corresponding set of size $s$. Invoking Lemma \ref{lm: randomGupta} and using conditional expectations, we get that there exists a terminal pair $(t,t') \in T_i$ such that
\begin{align*}
	\frac{\mathbb{E}_{\pi}[d_H(t,t')]}{d_G(t,t')} &~ \geq~  \frac{\mathbb{E}_{\pi}\left[d_H(t,t') ~|~ v \text{ is contracted}\right] \cdot 
	\mathbb{P}_{\pi} \left[v \text{ is contracted}\right]}{d_G(t,t')}  \\
	&~ \geq ~ 2\left(1-\frac{1}{s}\right)\left(1-\frac{s^2}{2c_1}\right) ~~ \geq ~~ 2 - \left(\frac{2}{s} + \frac{s^{2}}{c_1}\right),
\end{align*}
which can be made arbitrarily close to $2$ by setting $s$ and $c_1$ sufficiently large.
To be precise, given any $\epsilon > 0$, by setting $s = 5/\epsilon$ and $c_1 = 2 s^2/\epsilon$,
the distortion is at least $2-\epsilon$.
\end{pfof}




\section{Minor Construction for General Graphs}\label{sect:UB-general}
In this section we give  minor constructions that present numerous trade-offs between the distortion and size of DAMs.
Our results are obtained by combining the work of Coppersmith and Elkin~\cite{coppersmithE06} on sourcewise distance preservers with the well-known notion of spanners.

Given an undirected graph $G=(V,E, \ell)$ with terminals $T$, we let $\Pi_{u,v}$ denote the shortest path  between $u$ and $v$ in $G$.
Without loss of generality, we assume that for any pair of vertices $(u,v)$, the shortest path connecting $u$ and $v$ is \textit{unique}.
This can be achieved by slightly perturbing the original edge lengths of $G$ such that no paths have exactly the same length (see \cite{coppersmithE06}). The perturbation implies a \emph{consistent} tie-breaking scheme: whenever $\Pi$ is chosen as the shortest path, every subpath of $\Pi$ is also chosen as the shortest path.

For a graph $G$, let $N_G(u)$ denote the vertices incident to $u$ in $G$. We say that two paths $\Pi$ and $\Pi'$ branch at a vertex $u \in V(\Pi) \cap V(\Pi')$ iff $|N_{\Pi \cup \Pi'}(u)| > 2$. We call such a vertex $u$ a \textit{branching} vertex.
Let $\mathcal{P}$ denote the set of shortest paths corresponding to every pair of vertices in $G$. 
We review the following result proved in \cite[Lemma 7.5]{coppersmithE06}. 
\begin{lemma} \label{lemma: elkin} 
Any pair of shortest paths $\Pi, \Pi' \in \mathcal{P}$ has at most two branching vertices.
\end{lemma}

To simplify our exposition, we introduce the following definition.

\begin{definition}[Terminal Path Cover] Given $G = (V,E, \ell)$ with terminals $T$,
a set of shortest paths $\mathcal{P}' \subset \mathcal{P}$ is an $(\alpha, f(k))$-terminal path cover (abbr.~$(\alpha, f(k))$-TPc) of $G$ if
\begin{enumerate} 
\item $\mathcal{P}'$ covers the terminals, i.e. $T \subseteq V(H)$, where $H = \bigcup_{\Pi \in \mathcal{P}'} E(\Pi)$, 
\item  $|\mathcal{P}'| \leq f(k)$ and $\forall t,t' \in T$, $d_G(t,t') \leq d_H (t,t') \leq \alpha \cdot d_G(t,t')$.
\end{enumerate}
\end{definition}
We remark that the endpoints of the shortest paths in $\mathcal{P}'$ are not necessarily terminals. Now we give a simple algorithm generalizing the one presented by Krauthgamer et al.~\cite{distancepreserving}.

\begin{algorithm}[H]
\caption{\textsc{MinorSparsifier} (graph $G$, terminals $T$, $(\alpha,f(k))$-TPc $\mathcal{P}'$ of $G$)}
\begin{algorithmic}[1]
\State Set $H = \emptyset$. Then add all shortest paths from the path  cover $\mathcal{P}'$ to $H$.
\While{there exists a degree two non-terminal $v$ incident to edges $(v,u)$ and $(v,w)$} 
\State Contract the edge $(u,v)$, then set the length of edge $(u,w)$ to $d_H(u,w)$.
\EndWhile
\State \Return $H$
\end{algorithmic}
\label{algo: minor}
\end{algorithm}

The following lemma gives an upper bound on the size of the DAM output by Algorithm \ref{algo: minor}. It is an easy generalization of a lemma in \cite[Lemma 2.2]{distancepreserving}; for completeness, we give its proof in Appendix \ref{app:branching}.

\begin{lemma} \label{lemma: branching}
For a given graph $G = (V,E, \ell)$ with terminals $T \subset V$ and an $(\alpha,f(k))$-\emph{TPc} $\mathcal{P}'$ of $G$, \textsc{MinorSparsifier($G$,$T$,$\mathcal{P}'$)} outputs an $(\alpha,f(k)^2)$-\emph{DAM} of $G$. 
\end{lemma}

A trivial \textit{exact} terminal path cover for any $k$-terminal graph is to take the union of all terminal shortest paths,
which we refer to as the $(1,\calO(k^{2}))$-TPc $\mathcal{P}'$ of $G$.
Krauthgamer et al.~\cite{distancepreserving} used this $(1,\calO(k^{2}))$-TPc to construct an $(1,\mathcal{O}(k^{4}))$-DAM.
Here, we study the question of whether increasing the distortion slightly allows us to obtain a cover of size $o(k^2)$.
We answer this question positively, by reducing it to the well-known spanner problem.

Let $q \geq 1$ be an integer and let $G = (V,E,\ell)$ be an undirected graph. A $q$-spanner of $G$ is a subgraph $S = (V,E_S, \ell)$ such that
$\forall u,v \in V,~d_G(u,v) \leq d_S(u,v) \leq q \cdot d_G(u,v)$ .
We refer to $q$ and $|E_S|$ as the \textit{stretch} and \textit{size} of spanner $S$, respectively.
\begin{lemma}[\cite{AlthoferDDJS1993}] \label{lemma: spanner}
Let $q \geq 1$ be an integer. Any graph $G = (V,E,\ell)$ admits a $(2q-1)$-spanner $S$ of size $\mathcal{O}(|V|^{1+1/q})$.
\end{lemma}

We use the above lemma as follows. Given a graph $G=(V,E,\ell)$ with terminals $T$, we compute the complete graph $Q_T = (T,\binom{T}{2},d_G|T)$, where $d_G|T$ denotes the distance metric of $G$ restricted to the point set $T$ (In other words, for any pair of terminals $t,t' \in T$, the weight of the edge connecting them in $Q_T$ is given by $w_{Q_T}(t,t') = d_G(t,t')$). Recall that all shortest paths in $G$ are unique.

Using Lemma \ref{lemma: spanner}, we construct a $(2q-1)$-spanner $S$ of size $\mathcal{O} (k^{1+1/q})$ for $Q_T$. Observe that each edge of $S$ corresponds to an unique (terminal) shortest path in $G$ since $S$ is a subgraph of $Q_T$. Thus, the set of shortest paths corresponding to edges of $S$ form a $(2q-1, \mathcal{O} (k^{1+1/q}))$-TPc $\mathcal{P}'$ of $G$. Using $\mathcal{P}'$ with Lemma \ref{lemma: branching} gives the following theorem.

\begin{theorem}
Let $q \geq 1$ an integer. Any graph $G=(V,E,\ell)$ with $T \subset V$ admits a $(2q-1,\mathcal{O} (k^{2+2/q}))$-\emph{DAM}.
\end{theorem}
We mention two trade-offs from the above theorem. When $q=2$, we get an $(3,\mathcal{O} (k^3))$-DAM. When $q=\log k$, we get an $(\mathcal{O}(\log k),\mathcal{O}(k^2))$-DAM. These are new distortion-size trade-offs.

The above method allows us to have improved guarantees for bounded treewidth graphs.

\begin{theorem}\label{thm:treewidth}
Let $q \geq 1$ be an integer. Any graph $G=(V,E,\ell)$ with treewidth at most $p$, $T \subset V$ and $k \geq p$
admits a $(2q-1,\mathcal{O}(p^{1+2/q} k))$-\emph{DAM}.
\end{theorem}
We defer the proof of the above theorem to Appendix \ref{app:treewidth}. 
The theorem implies, in particular, that any graph $G$ with treewidth at most $p$ admits an $(\mathcal{O}(\log p),\mathcal{O}(p k))$-DAM.

\section{Minor Construction for Graphs Excluding a Fixed Minor}
In this section we give improved guarantees for distance approximating minors for special families of graphs. Specifically, we show that graphs that exclude a fixed minor admit an $(\mathcal{O}(1),\widetilde{\mathcal{O}}(k^{2}))$-DAM. This family of graphs includes, among others, the planar graphs. 

The reduction to spanner in Section \ref{sect:UB-general} does not consider the structure of $Q_T$, which is inherited from the input graph.
We exploit this structure, together with the use of the randomized Steiner Point Removal Problem,
which is equivalent to finding an $(\alpha, 0)$-rDAM. 


We start by reviewing the following result of Englert et al.~\cite{englert10},
which shows that for graphs that exclude a fixed minor, there exists a randomized minor with constant distortion.
\begin{theorem}[\cite{englert10}, Theorem 14]
\label{thm2}
Let $\alpha = \calO(1)$. Given a graph that excludes a fixed minor $G = (V,E, \ell)$ with $T \subset V$,
there is a probability distribution $\pi$ over minors $H = (T,E', \ell')$ of $G$,
such that $\forall \; t,t' \in T, ~\mathbb{E}_{H\sim\pi} [d_H(t,t')] \leq \alpha \cdot d_G(t,t')$ and for every minor $H$ in the support of $\pi$, $d_H(t,t') \geq d_G(t,t')$.
\end{theorem}

Given a graph $G$ that excludes a fixed minor, any minor $H$ of $G$ only on the terminals also excludes the fixed minor. Thus $H$ has $\calO(k)$ edges~\cite{Thomason1984}. This leads to the corollary below.

\begin{corollary}
\label{cor1} Let $\alpha = \calO(1)$. Given a graph that excludes a fixed minor $G = (V,E, \ell)$ with $T \subset V$ and $Q_T$ as defined in Section 4,
there exists a probability distribution $\pi$ over subgraphs $H=(T,E', \ell')$ of $Q_T$, each having at most $\mathcal{O}(k)$ edges, such that for all
$t,t' \in T,~\mathbb{E}_{H\sim\pi}[d_H(t,t')] \leq \alpha \cdot d_{Q_T}(t,t')$.
\end{corollary}
\begin{proof}
Let $\pi$ be the distribution over minors of $G$ from Theorem \ref{thm2}, then every minor in its support is clearly a subgraph of $Q_T$ with $\calO(k)$ edges. Since during the construction of these minors we may assume that $\forall (t,t') \in E',~\ell'(t,t') = d_G(t,t')$, the corollary follows.
\end{proof}

\begin{lemma} \label{lemma: planar 1}
Given a graph that excludes a fixed minor $G = (V,E, \ell_G)$ with $T\subset V$, and $Q_T$ as defined in Section 4, there exists an $\calO(1)$-spanner $S$ of size $\mathcal{O} (k \log k)$ for $Q_T$.
\end{lemma}
\begin{proof}
Let $\pi$ be the probability distribution over subgraphs $H$ from Corollary \ref{cor1}. Set $S = \emptyset$.
First, we sample independently $q = 3 \log k$ subgraphs $H_1, \ldots, H_q$ from $\pi$.
We then add the edges from all these subgraphs to the graph $S$, i.e., $E_S = \bigcup_{i=1}^{q} E_{H_i}$.
Fix an edge $(t,t')$ from $Q_T$ and a subgraph $H_i$. By Corollary \ref{cor1} and the Markov inequality,
$\mathbb{P} [d_{H_i}(t,t') \geq 2 \alpha \cdot d_{Q_T}(t,t')] \leq  2^{-1}$, and hence\vspace*{-0.15cm}
$$
	\mathbb{P}[d_{S}(t,t') \geq 2\alpha \cdot d_{Q_T}(t,t')]  = \prod_{i=1}^{q} \mathbb{P}[d_{H_i}(t,t') \geq 2 \alpha \cdot d_{Q_T}(t,t')]  \leq 2^{-q} = k^{-3}.\vspace*{-0.15cm}
$$
Applying of the union bound overall all edges from $Q_T$ yields\vspace*{-0.15cm}
$$
	\mathbb{P}[\text{there exists an edge } (t,t') \in Q_T \text{ s.t. } d_{S}(t,t') \geq 2\alpha \cdot d_{Q_T}(t,t')] \leq k^{2} \cdot k^{-3} = k^{-1}.
$$
Hence, for all edges $(t,t')$ from $Q_T$, with probability at least $1 - 1/k$, we preserve the shortest path distance between $t$ and $t'$
up to a factor of $2 \alpha = \calO(1)$ in $S$. Since $S$ is a subgraph of $Q_T$,
this implies that there exists a $\calO(1)$-spanner $S$ of size $\calO(k \log k)$ for $Q_T$.
\end{proof}
Similar to the last section, the set of shortest paths corresponding to edges of $S$ is an $(\calO(1), \calO(k \log k))$-TPc  $\mathcal{P}'$ of $G$.
Using  $\mathcal{P}'$ with Lemma \ref{lemma: branching} gives the following theorem.
\begin{theorem}
Any graph that excludes a fixed minor $G = (V,E, \ell)$ with $T \subset V$ admits an $(\mathcal{O} (1),\widetilde{\mathcal{O}} (k^{2}))$-\emph{DAM}.
\end{theorem}

\section{Minor Construction for Planar Graphs}\label{sect:UB-planar}
In this section, we show that for planar graphs, we can  improve the constant guarantee bound on the
distortion to $3$ and $1+\varepsilon$, respectively,
without affecting the size of the minor. Our work builds on existing techniques used in the context of approximate distance oracles,
thereby bypassing our previous spanner reduction. Both results use essentially the same ideas
and rely heavily on the fact that planar graphs admit separators with special properties.

We say that a graph $G=(V,E,\ell)$ admits a $\lambda$-separator if there exists a set $R \subseteq V$
whose removal partitions $G$ into connected components, each of size at most $\lambda n$, where $1/2 \leq \lambda <1$.
Lipton and Tarjan~\cite{lipton79} showed that every planar graph has a $2/3$-separator $R$ of size $\mathcal{O} (\sqrt{n})$.
Later on, Gupta et al.~\cite{guptaKR04} and Thorup~\cite{thorup04} independently observed that one can modify their construction
to obtain a $2/3$-separator $R$, with the additional property that $R$ consists of vertices belonging to shortest paths from $G$
(note that this $R$ is not guaranteed to be small). We briefly review the construction of such \textit{shortest path separators}. 

Let $G=(V,E,\ell)$ be a triangulated planar graph (the triangulation is guaranteed by adding infinity edge lengths among the missing edges). Further, let us fix an arbitrary shortest path tree $A$ rooted at some vertex $r$. Then, it can be inferred from the work of Lipton and Tarjan~\cite{lipton79} that there always exists a non-tree edge $e=\{u,v\}$ of $A$ such that the fundamental cycle $\mathcal{C}$ in $A \cup \{e\}$, formed by adding the non-tree edge $e$ to $A$, gives a $2/3$-separator for $G$. Because $A$ is a tree, the separator will consist of two paths from the $\text{lca}(u,v)$ to $u$ and $v$. We denote such paths by $P_1$ and $P_2$, respectively. Both paths are shortest paths as they belong to $A$. We will show how to use such separators to obtain terminal path covers. Before proceeding, we give the following preprocessing step.

\medskip

\noindent \textbf{Preprocessing Step. }Given a planar graph $G = (V,E,\ell)$ with $T \subset V$, the algorithm \textsc{MinorSparsifier}($G$, $T$, $\mathcal{P}'$) with $\mathcal{P}'$ being the $(1,\mathcal{O}(k^{2}))$-TPc of $G$, produces an $(1,\mathcal{O}(k^{4}))$-DAM $G'$ for $G$. To simplify our notation, we will use $G$ instead of $G'$ in the following, i.e., we assume that $G$ has at most $\mathcal{O}(k^{4})$ vertices.


\subsection{Stretch-$3$ Guarantee}

When solving a graph problem, it is often that the problem can be more easily solved for simpler graph instances, e.g., trees.
Driven by this, it is desirable to reduce the problem from arbitrary graphs to one or several tree instances, possibly allowing a small loss in the quality of the solution. Along the lines of such an approach, Gupta et al.~\cite{guptaKR04} gave the following definition in the context of shortest path distances.
\begin{definition}[Forest Cover] Given a graph $G=(V,E,\ell)$, a forest cover (with stretch $\alpha$) of $G$ is a family $\mathcal{F}$ of subforests $\{F_1,F_2,\ldots,F_k\}$ of $G$ such that for every $u,v \in V$, there is a forest $F_i \in \mathcal{F}$ such that $d_{G}(u,v) \leq d_{F_i}(u,v) \leq \alpha \cdot d_G(u,v)$.
\end{definition}

If we restrict our attention to planar graphs, Gupta et al.~\cite{guptaKR04} used shortest path separators (as described above) to give a divide-and-conquer algorithm for constructing forest covers with small guarantees on the stretch and size. Here, we slightly modify their construction for our purpose.
Before proceeding to the algorithm, we give the following useful definition.
\begin{definition}
Let $t$ be a terminal and let $P$ be a shortest path in $G$. Then $t_{\min}^{P}$ denotes the vertex of $P$ that minimizes $d_G(t,p)$, for all $p \in P$, breaking ties lexicographically. 
\end{definition}

\begin{algorithm}[H]
\caption{\textsc{ForestCover} (planar graph $G$, terminals $T$)}
\begin{algorithmic}[1]
\If{$|V(G)| \leq 1$}
	\Return $V(G)$
\EndIf
\State Compute a $2/3$-separator $\mathcal{C}$ consisting of shortest paths $P_1$ and $P_2$ for $G$.
\For{$i=1,2$} 
\State Contract $P_i$ to a single vertex $p_i$ and compute a shortest path tree $L_i$ from $p_i$.
\State Expand back the contracted edges in $L_i$ to get the tree $L_i'$.
\For{every terminal $t \in T$}
 \State Add $t_{\min}^{P_i}$ as a terminal in the tree $L_i'$.
\EndFor
\EndFor
\State Let $(G_1,T_1)$ and $(G_2,T_2)$ be the resulting connected graphs from $G \setminus \mathcal{C}$, 
\Statex where $T_1$ and $T_2$ are disjoint subsets of the terminals $T$ induced by $\mathcal{C}$.
\Statex \texttt{// Note that all distances involving terminals from $\mathcal{C}$ are taken care of.}
\State \Return $\bigcup_{i=1}^{2} L_i' \cup \bigcup_{i=1}^{2} \textsc{ForestCover}(G_i,T_i)$.
\end{algorithmic}
\label{algo: forestcover}
\end{algorithm}

\begin{algorithm}[H]
\caption{\textsc{PlanarTPc-1} (planar graph $G$, terminals $T$)}
\begin{algorithmic}[1]
\State Set $\mathcal{P}' = \emptyset$. Set $\mathcal{F} = \text{\textsc{ForestCover}}(G,T)$. 
\For{every forest $F_i \in \mathcal{F}$} 
\State Let $R_i$ be the terminal set of $F_i$ and let $\mathcal{P}'_i$  be the (trivial) $(1,\mathcal{O}(k^{2}))$-TPc of $F_i$;
\State Compute $F_i'$ = \textsc{MinorSparsifier}($F_i$, $R_i$, $\mathcal{P}_i'$).
\State Add the shortest paths corresponding to the edges of $F_i'$ to $\mathcal{P}'$.
\EndFor
\State \Return $\mathcal{P}'$
\end{algorithmic}
\label{algo: minor1}
\end{algorithm}

Gupta et al.~\cite{guptaKR04} showed the following guarantees for Algorithm \ref{algo: forestcover}.

\begin{theorem}[\cite{guptaKR04}, Theorem 5.1] \label{thm: guptacover}
Given a planar graph $G=(V,E,\ell)$ with $T \subset V$, \textsc{ForsetCover}$(G,T)$ produces a stretch-$3$ forest cover with $\mathcal{O} (\log |V|)$ forests. 
\end{theorem} 

We note that the original construction does not consider terminal vertices, but this does not worsen neither the stretch nor the size of the cover. The only difference here is that we need to add at most $k$ new terminals to each forest compared to the original number of terminals in the input graph. This modification affects our bounds on the size of a minor only by a constant factor. 

Below we show that using the above theorem one can obtain terminal path covers for planar graphs. 

\begin{lemma} Given a planar graph $G = (V,E, \ell)$ with $T \subset V$, \textsc{PlanarTPc-1}$(G,T)$ produces an $(3,\mathcal{O} (k \log k))$-\emph{TPc} $\mathcal{P}'$ for $G$. 
\end{lemma}
\begin{proof}
We first review the following simple fact, whose proof can be found in \cite{distancepreserving}.
\begin{fact} \label{fact: treecover}
Given a forest $F=(V,E,\ell)$ with terminals $T \subset V$ and $\mathcal{P}'$ being the (trivial) $(1,\mathcal{O}(k^{2}))$-\emph{TPc} of $F$, the procedure \textsc{MinorSparsifier}$(F,T,\mathcal{P}')$ outputs an $(1,k)$-\emph{DAM}.  
\end{fact}

Let us proceed with the analysis. Observe that from the Preprocessing Step our input graph $G$ has at most $\mathcal{O}(k^{4})$ vertices. Thus, applying Theorem \ref{thm: guptacover} on $G$ gives a stretch-$3$ forest cover $\mathcal{F}$ of size $\mathcal{O}(\log k)$. In addition, recall that all shortest paths are unique in $G$.

Next, let $F_i$ by any forest from $\mathcal{F}$. By construction, we note that each tree belonging to $F_i$ has the nice property of being a concatenation of a given shortest path with another shortest path tree. We will exploit this in order to show that every edge of the minor $F_i'$ for $F_i$ corresponds to the (unique) shortest path between its endpoints in $G$.


To this end, let $e' = (u,v)$ be an edge of $F_i'$ that does not exist in $F_i$. Since $F_i'$ is a minor of $F_i$, we can map back $e'$ to the path $\Pi_{u,v}$ connecting $u$ and $v$ in $F_i$. Because of the additional terminals $u_{\min}^{P_i}$ added to $F_i$, we claim that $\Pi_{u,v}$ is entirely contained either in some shortest path tree $L_j$ or some shortest path separator $P_j$. Using the fact that subpaths of shortest paths are shortest paths, we conclude that the length of the path $\Pi_{u,v}$ (or equivalently, the length of edge $e'$) corresponds to the unique shortest path connecting $u$ and $v$ in $G$. The same argument is repeatedly applied to every such edge of $F_i'$.

By construction we know that $F_i$ has at most $2k$ terminals. Using Fact \ref{fact: treecover} we get that $F_i'$ contains at most $4k$ edges. Since there are $\mathcal{O} (\log k)$ forests, we conclude that the terminal path cover $\mathcal{P}'$ consists of $\mathcal{O} (k \log k)$ shortest paths. The stretch guarantee follows directly from that of cover $\mathcal{F}$, since $F_i'$ exactly preserves all distances between terminals in $F_i$.
\end{proof}

\begin{theorem}
Any planar graph $G=(V,E,\ell)$ with $T \subset V$ admits a $(3, \widetilde{\mathcal{O}}(k^{2}))$-\emph{DAM}.
\end{theorem}


\subsection{Stretch-$(1+\varepsilon)$ Guarantee}

Next we present our best trade-off between distortion and size of minors for planar graphs. Our idea is to construct terminal path covers using the construction of Thorup~\cite{thorup04} in the context of approximate distance oracles in planar graphs. Here, we modify a simplified version due to Kawarabayashi et al.~\cite{sommer11}. 

The construction relies on two important ideas. Similarly to the stretch-$3$ result, the first idea is to recursively use shortest path separators to decompose the graph. The second consists of approximating shortest paths that cross a shortest path separator. Below we present some necessary modification to make use of such a construction for our purposes.

Let $P$ be a shortest path in $G$. For a terminal $t \in T$, we let the pair $(p,t)$, where $p \in P$, denote the \textit{portal} of $t$ with respect to the path $P$. An $\varepsilon$-cover $C(t,P)$ of $t$ with respect to $P$ is a set of portals with the following property:
$$
	\forall p \in P, ~ \exists q \in C(t,P) \text{  s.t.  } d_G(t,q) + d_G(q,p) \leq (1+\varepsilon)d_G(t,p)
$$
Let $(t,t')$ by any terminal pair in $G$. Let $\Pi_{t,t'}$ be the (unique) shortest path that crosses the path $P$ at vertex $w$. Then using the $\varepsilon$-covers $C(t,P)$ and $C(t',P)$, there exist portals $(t,p)$ and $(p',t')$ such that the new distance between $t$ and $t'$ is
\begin{align} \label{eqn: stretch}
\begin{split}
	d_G(t,p) + d_G(p,p') + d_G(p',t') &~ \leq~ d_G(t,p) + d_G(p,w) + d_G(w,p') + d_G(p',t') \\
	&~ \leq~ (1+\varepsilon) d_G(t,t')
\end{split}
\end{align}
The new distance clearly dominates the old one. The next result due to Thorup~\cite{thorup04} shows that maintaining a small number of portals per terminal suffices to approximately preserve terminal shortest paths.

\begin{lemma} \label{lemm: portal} 
Let $\varepsilon > 0$. For a given terminal $t \in T$ and a shortest path $P$, there exists an $\varepsilon$-cover $C(t,P)$ of size $\mathcal{O}(1/\varepsilon)$.
\end{lemma}

The above lemma leads to the following recursive procedure.

\begin{algorithm}[H]
\caption{\textsc{PlanarTPc-2} (planar graph $G$, terminals $T$)}
\begin{algorithmic}[1]
\If{$|V(G)| \leq 1$ or $T = \emptyset$}
	\Return $\emptyset$;
\EndIf
\State Set $\mathcal{B} = \emptyset$;
\State Compute a $2/3$-separator $\mathcal{C}$ consisting of shortest paths $P_1$ and $P_2$ and add them to $\mathcal{B}$.
\For{every terminal $t \in T$} 
\State Compute $\varepsilon$-covers $C(t,P_1)$ and $C(t,P_2)$.
\For{every portal $(t,p) \in C(t,P_1) \cup C(t,P_2)$}
\State Add the shortest path $\Pi_{t,p}$ to $\mathcal{B}$.
\EndFor
\EndFor
\State Let $(G_1,T_1)$ and $(G_2,T_2)$ be the resulting connected graphs from $G \setminus \mathcal{C}$, 
\Statex where $T_1$ and $T_2$ are disjoint subsets of the terminals $T$ induced by $\mathcal{C}$.
\Statex \texttt{// Note that all distances involving terminals from $\mathcal{C}$ are taken care of.}
\State \Return $\mathcal{B} \cup \bigcup_{i=1}^{2} \textsc{PlanarTPc-2}(G_i,T_i)$.
\end{algorithmic}
\label{algo: minor2}
\end{algorithm}

\begin{lemma}\label{lem:tpc-1-plus-ep}
Given a planar graph $G = (V,E, \ell)$ with $T \subset V$, \textsc{PlanarTPc-2}$(G,T)$ outputs an $(1+\varepsilon,\mathcal{O} (k \log k / \varepsilon))$-\emph{TPc} $\mathcal{P}'$ for $G$. 
\end{lemma}

\begin{proof}
From the Preprocessing Step we know that $G$ has at most $\mathcal{O}(k^{4})$ vertices. Further,  recall that removing the vertices that belong to the shortest path separators from $G$ results into two graphs $G_1$ and $G_2$, whose size is at most $2/3 \cdot |G|$. Thus, there are at most $\mathcal{O}(\log k)$ levels of recursion for the above procedure.  

Let $\mathcal{P}'$ be the terminal path cover output by \text{PlanarTPc-2}$(G,T)$. We first bound the number of separator shortest paths added in Step 3. Note that at any level of the recursion there at most $k$ terminals and, thus the number of recursive calls per level is at most $k$. Since we added two paths per recursive call, we get that there are at most $\mathcal{O} (k \log k)$ paths overall.

We now continue with the counting or portals. Let $t \in T$ be any terminal and consider any recursive call applied on the current graph $(G',T')$. If $t \not\in T'$, then we simply ignore $t$. Otherwise, $t$ either belongs to one of the separator shortest paths in $G'$ or one of the partitions induced by the separators. In the first case, we know that $t$ is retained because we added $P_1$ and $P_2$ to $\mathcal{P'}$ and these are already counted. In the second case, using Lemma \ref{lemm: portal}, we add $\mathcal{O}(1/\varepsilon)$ shortest paths connecting portals from $C(t,P_1)$ and $C(t,P_2)$. Therefore, in any recursive call, we maintain at most $\mathcal{O}(1/\varepsilon)$ shortest paths per terminal. Since every terminal can participate in at most $\mathcal{O}(\log k)$ recursive calls, we get that the total number of portal-shortest paths is at most $\mathcal{O} (k \log k / \varepsilon)$. Combining both bounds, it follows that the size of $\mathcal{P}'$ is at most $\mathcal{O}(k \log k / \varepsilon)$. 

It remains to show the stretch guarantee of $\mathcal{P}'$. Let $R$ be the recursion tree of the algorithm, where every node corresponds to a recursive call. For any pair $t,t' \in T$, let $a \in V(R)$ associated with $(G_a, T_a)$ be the leafmost node such that $t,t' \in T_a$. Then, it follows that among all ancestors of $a$ in the tree $R$, there must exist a separator path $P_i$, $i=1,2$ that crosses $\Pi_{t,t'}$ and attains the minimum length. The stretch guarantee follows directly from (\ref{eqn: stretch}).
\end{proof}

\begin{theorem}\label{thm:planar-oneplus}
Any planar graph $G=(V,E,\ell)$ with $T \subset V$ admits an $(1+\varepsilon, \widetilde{\mathcal{O}}((k/\varepsilon)^2)$-\emph{DAM}.
\end{theorem} 

\section{Discussion and Open Problems}

We note that there remain gaps between some of the best upper and lower bounds, e.g.,
for general graphs and distortion $3-\epsilon$, the lower bound is $\Omega(k^{6/5})$,
while for distortion $3$, our upper bound is $\calO(k^3)$.
Improving the bounds is an interesting open problem.


Our techniques for showing upper bounds rely heavily on the spanner reduction.
For planar graphs, Krauthgamer et al.~\cite{distancepreserving} showed that to achieve distortion $1+o(1)$, $\Omega(k^2)$ non-terminals are needed;
we bypass the spanner reduction to construct an $(1+\varepsilon,\widetilde{\mathcal{O}}(k/\varepsilon)^{2})$-DAM,
which is tight up to a poly-logarithmic factor.
It is an interesting open question on whether similar guarantees can be achieved for general graphs.

\section*{Acknowledgements}
The authors thank Veronika Loitzenbauer and Harald R{\"{a}}cke for the helpful discussions.

The research leading to these results has received funding from the European Research Council under the European Union's Seventh Framework Programme (FP/2007-2013) / ERC Grant Agreement no.~340506. The research leading to these results has received funding from the European Research Council under the European Union's Seventh Framework Programme (FP/2007-2013) under grant agreement no.~317532.

\printbibliography[heading=bibintoc] 

 \infull{\clearpage
 \section*{Appendix}
 \appendix
\section{Missing Proofs in Section \ref{sect:lower-bound}} \label{app:basic-of-minor}

Let $G$ be an output graph from the black-box.
In any minor $H$ of $G$, we say a super-node is of Type-A if $S(u)$ contains only non-terminals in $G$;
any other super-node $u$, for which $S(u)$ contains \emph{exactly} one terminal, is of Type-B. Here are two simple facts:
\begin{enumerate}
\item[(a)] If $u$ is of Type-A, since $G[S(u)]$ is connected, the non-terminals in $S(u)$ must belong to the same group.
\item[(b)] If $u$ is of Type-B, let $t$ be the terminal in $S(u)$. If $S(u)$ contains a vertex from some group $R$, then $t\in R$.
\end{enumerate}

%
%
%

\begin{pfof}{Lemma \ref{lem:unique-group-for-edge}}
Existence of $R$ is easy to prove by a simple induction on the minor operation sequence that generates $H$ from $G$.
To show the uniqueness, we proceed to a case analysis.
In the first case, either $u_1$ or $u_2$ is of Type-A. Then the uniqueness is trivial by fact (a).

In the second case, both $u_1,u_2$ are of Type-B. For $i=1,2$, let $t_i$ be the terminal in $S(u_i)$.
Suppose there are two groups $R_a,R_b$ that intersect both $S(u_1)$ and $S(u_2)$.
Then by fact (b), $t_1,t_2$ are in both $R_a$ and $R_b$, a contradiction.
\end{pfof}

\begin{pfof}{Lemma \ref{lem:interchange-at-terminal}}
Since $S(u_2)$ contains vertices from both $R_1$ and $R_2$, $u_2$ must be of Type-B, i.e., $S(u_2)$ contains exactly one terminal $t$.
By fact (b), $t$ is in both $R_1$ and $R_2$.
\end{pfof}

\begin{pfof}{Lemma \ref{lm: randomGupta}}
Due to the the simple structure of $S$, we can easily observe that there are exactly $k$ different minors of $S$. Specifically, for every $t \in T$, let $H_t$ be the star graph (now only on terminals) obtained by contracting the edge $(w,t)$ in $G$. Further, let $\pi$ be some probability distribution on $\mathcal{H} = \{H_t : t \in T\}$. The expected distortion for embedding $S$ into $\pi$ is
$$
	\max_{t',t'' \in T} \frac{\mathbb{E}_{\pi}\left[{d_{H_t}(t',t'')}\right]}{d_{S}(t',t'')}
$$
Let us have a closer look at the above relation. First, note that for every $t',t'' \in T$, $d_S(t',t'') = 2$ and every edge of $H_t$ must set weights of size $2$ to all of its edges because of the domination property of terminal distances. This implies that terminal pairs that are connected in some $H_t$ do not suffer any distortion, while those that are not connected suffer a distortion of $2$. Furthermore, for any pair $t',t'' \in T$, the probability that $t'$ and $t''$ are connected is $p_{t'} + p_{t''}$. Combining the above facts we get that for every $t',t'' \in T$
$$
	 \frac{\mathbb{E}_{\pi}\left[{d_{H_t}(t',t'')}\right]}{d_{S}(t',t'')} =  \pi_{t'} + \pi_{t''} + \sum_{t \in T \setminus \{t,t'\}} 2\pi_{t} = D(t',t'')
$$

Now, we use a standard result due to Charikar et al.~\cite{charikar98} that computes the probability distribution that minimizes the expected distortion using the following linear program
\begin{center}
\begin{tabular}{rcll}
$\min$ & $\lambda$  & & \\
\rule{0pt}{3ex}$D(t',t'')$ & $\leq$ & $\lambda$ & $\forall t',t'' \in T$   \\
\rule{0pt}{3ex}$\displaystyle \sum_{H_{t} \in \mathcal{H}} \pi_{t}$ & $=$ & $1$ & \\
\rule{0pt}{2ex}$\pi_t$ & $\geq$ & $0$ & $ \forall H_{t} \in \mathcal{H}$ 
\end{tabular}
\end{center}
By an easy inspection, one can observe that the above LP is symmetric, i.e any permutation of the variables $\pi_t$ leaves the feasible region unchanged. Consequently, invoking a known result about symmetric LPs (see~\cite{bodi}), we get that there exists an optimal solution with all $\pi_t$'s being equal to each other. The latter implies that for every $t \in T$, we have $\pi_t = 1/k$, which in turn gives that $\lambda = 2(1-1/k)$, what we wanted to show. 
\end{pfof}
\subsection{A note on contraction-based Vertex Flow Sparsifiers}
A closely related concept to Distance Approximating Minors is Vertex Flow Sparsification. Roughly speaking, given a large network $G$ with some specified terminal set $T$, we want to construct a smaller graph $H$ that contains the terminal set $T$ and preserves all multi-commodity flows from $G$ up to some factor $q \geq 1$. In contrast to distance sparsifiers, here we do not pose any constraint about the structure of the sparsifier $H$.

However, in the setting where the sparsifiers lies only on the terminals, i.e. $V(H) = T$, it is customary to construct sparsifiers that are convex combination of minors of $G$ (see~\cite{englert10}). This leads to the concept of contraction-based vertex flow sparsifiers, i.e. sparsifiers that are convex combination of graphs that were obtained by performing edge contractions in $G$.

We note that a similar lower bound to that of Theorem \ref{thm:rand-lower-bound-star} can be obtained for contraction-based vertex flow sparsifiers. The only modification needed is Lemma \ref{lm: randomGupta}. We state its analogue below. 
\begin{lemma}
Let $S=(T \cup \{v\},E)$ be an unweighted star with $k \geq 3$ terminals. Then, for any probability distribution over minors of $S$,
there exists an edge in $S$ with load at least $2(1-1/k)$.
\end{lemma}

The proof of the above Lemma appears in~\cite{gramoz}. However, we omit further details since this is beyond the scope of this work. 

\section{Lower Bounds for SPR Problem on Trees} \label{app:explicit-tree-value}

Chan et al.~\cite{chan} considered unweighted complete binary tree with height $h$,
and showed that as $h\nearrow \infty$, the minimum distortion of SPR problem tends to $8$.
However, it is not clear from their proof how the minimum distortion depends on $h$,
which is needed for Theorem \ref{thm:lower-bound-distortion-25}.
In this section, we use their ideas on unweighted complete ternary trees to derive such a dependence.

\newcommand{\drl}{\text{\textsf{DRL}}}
\newcommand{\cs}{\mathcal{S}}

Let $T_h$ denote a unweighted complete ternary tree of height $h$, where the leaves are the terminals.
Let $\cs_h$ denote the collection of all minors of $T_h$.
For each of its node $u$, let $T(u)$ denote the sub-tree rooted at $u$, and let $t(u)$ denote the terminal which $u$ contracts into.
Denote the root by $r$, and its three children by $x,y,z$.
Without loss of generality, we assume that $r$ is contracted into a terminal $t_r$ in $T(x)$, i.e., $t(r) = t_r$.
Then, let\footnote{Formally speaking,
there can be infinitely many minors (with weights) of $T_h$ with distortion at most $\alpha$,
so we should use $\inf$ instead of $\min$ in the definition.
Yet, for each fixed minor without weight, the standard restriction~\cite[Definition 1.3]{KammaKN2015} is the optimal weight assignment.
Since there are only finitely many minors of $T_h$ (without weights), we can replace $\inf$ by $\min$.}
$$\drl(h,\alpha) := \min_{H\in \cs_h,\text{ distortion }\leq \alpha}~~\max_{\text{terminal }t\in T(y)\cup T(z)}~~d_H(t_r,t).$$
If there is not such a minor $H$, then $\drl(h,\alpha) = +\infty$ by default.
Note that when $\alpha$ increases, $\drl(h,\alpha)$ decreases.

\begin{figure}[h]
\begin{center}
\includegraphics[scale=0.77]{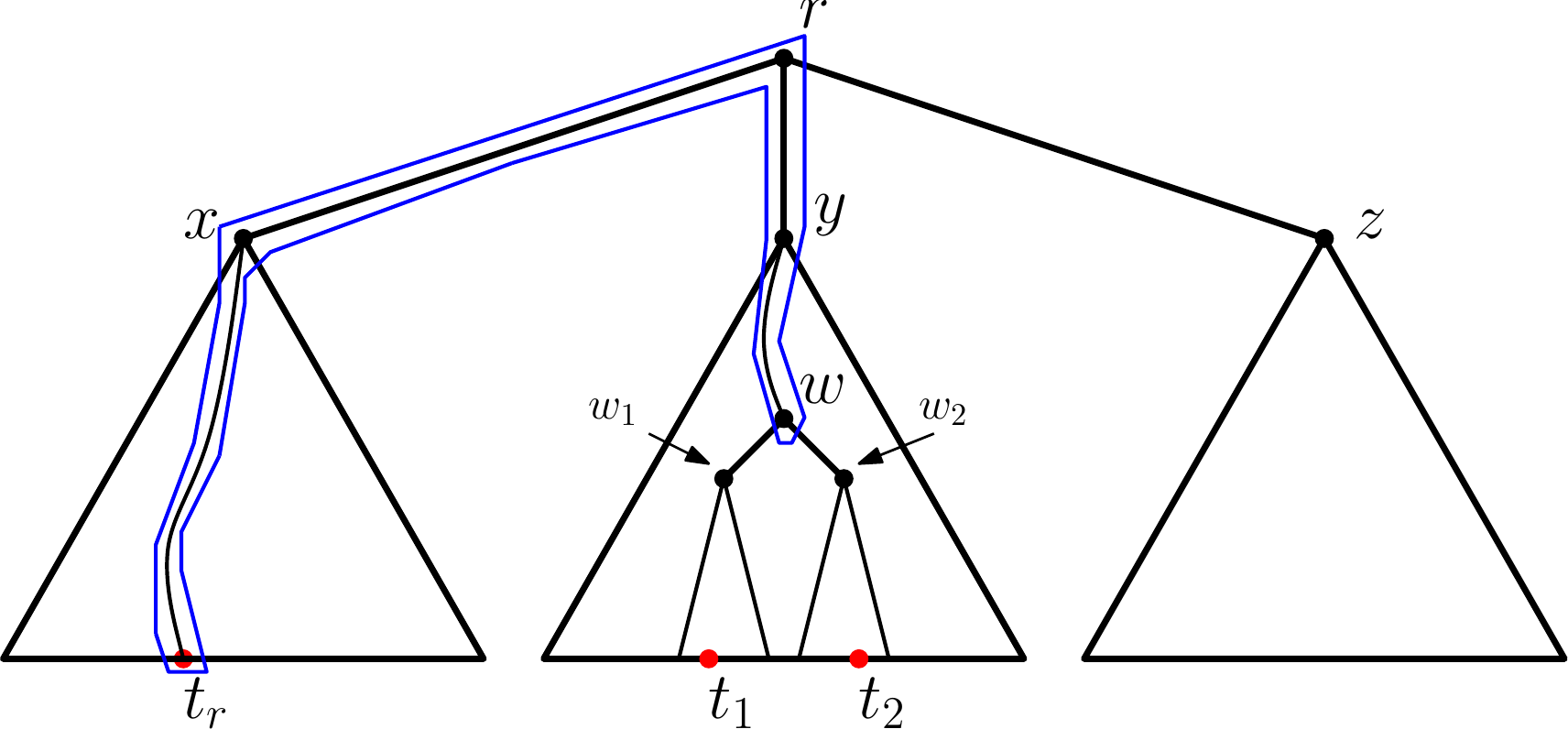}
\caption{The definitions of nodes in $T_h$. The blue polygon includes all nodes in $S(t_r) = S(r) = S(w)$.}
\label{fig:tree}
\end{center}
\end{figure}

Let $H\in\cs_h$ be a minor of $T_h$ with distortion $\leq \alpha$.
Let $w$ denote a deepest node in $T(y)\cup T(z)\cup \{r\}$ such that $t(w) = t_r$.
Let $\ell$ be the distance between $r$ and $w$ in $T_h$.
Let $w_1,w_2$ be two children of $w$ which are not in $T(x)$.
See Figure \ref{fig:tree} above for a picture of all the above definitions.

Then, by the definition of $\drl$, there exists two terminals $t_1\in T(w_1)$ and $t_2\in T(w_2)$ such that for $i=1,2$,
$d_H(t_i,t(w_i)) \geq \drl(h-\ell-1,\alpha)$.
Also, for $i=1,2$, $d_H(t(w_i),t_r) \geq d_{T_h}(t(w_i),t_r) = 2h$.
Hence,
\begin{align*}
	d_H(t_1,t_2) &~ =~  d_H(t_1,t(w_1)) + d_H(t(w_1),t_r) + d_H(t_r,t(w_2)) + d_H(t(w_2),t_2)  \\
	&~ \geq ~ 2\left[\drl(h-\ell-1,\alpha) + 2h\right].
\end{align*}

Recall that $d_{T_h}(t_1,t_2) = 2(h-\ell)$.
Hence, the distortion w.r.t.~$t_1,t_2$ is at least
$$\frac{\drl(h-\ell-1,\alpha) + 2h}{h-\ell}.$$
This quantity cannot be larger than $\alpha$.

We are ready to give a recurrence relation that bounds $\drl(h,\alpha)$ from below:
\begin{equation}\label{eq:drl-recur}
\drl(h,\alpha) \geq \min_{\ell\in [0,h-1]:~\frac{\drl(h-\ell-1,\alpha) + 2h}{h-\ell} \leq \alpha} \drl(h-\ell-1,\alpha) + 2h,
\end{equation}
while the initial conditions are: $\forall \alpha\geq 1,~\drl(0,\alpha) = 0$, and
$$\drl(1,\alpha) =
\begin{cases}
+\infty, & \text{if }\alpha < 2;\\
2, & \text{if }\alpha \geq 2.
\end{cases}$$

Let $\alpha_h$ denote the minimum distortion of $T_h$.
By letting $\ell$ runs over all possible distances between $r$ and $w$, we obtain the following lower bound on $\alpha_h$:
\begin{equation}\label{eq:alpha-recur}
\alpha_h \geq \min_\alpha~~\max\left\{\alpha , \left(\min_{\ell\in [0,h-1]} \frac{\drl(h-\ell-1,\alpha) + 2h}{h-\ell}\right)\right\}.
\end{equation}

We compute the lower bounds in \eqref{eq:drl-recur} and \eqref{eq:alpha-recur} using math software.
In the table below, we give the lower bounds on $\alpha_h$ for $h\in [3,10]$ and $h=1000$.

\smallskip

\begin{center}
\begin{tabular}{|c||c|c|c|c|c|c|c|}
\hline
$\mathbf{h}$ & 2 & 3,4 & 5 & 6,7 & 8 & 9,10 & 1000\\
\hline
$\mathbf{\alpha_h}$ & 3 & 4 & $22/5 = 4.40$ & $14/3 \approx 4.66$ & 5 & $26/5 = 5.20$ & $257/35 \approx 7.34$\\
\hline
\end{tabular}
\end{center}
\section{Missing Proofs in Section \ref{sect:UB-general}} \label{app:branching}

\begin{pfof}{Lemma \ref{lemma: branching}}
First, it is clear that the union over paths of $\mathcal{P}' \subset \mathcal{P}$ is a minor of $G$ (this can be alternatively viewed as deleting non-terminals and edges that do not participate in any of the shortest paths in $\mathcal{P}'$). Further, the algorithm performs only edge contractions. Thus, the produced graph $H$ is a minor of $G$. 

Since contracting edges incident to non-terminals of degree two does not affect any distance in $H$, the distortion guarantee follows directly from that of the cover $\mathcal{P}'$. Thus, it only remains to show the bound on the size of $H$.

To this end, consider any two paths $\Pi, \Pi'$ from $\mathcal{P}'$. From Lemma \ref{lemma: elkin}, we know that $\Pi$ and $\Pi'$ branch in at most two vertices. Let $u_1$ and $u_2$ denote such vertices. Due to the tie-breaking scheme in $G$, we know that the shortest path $\Pi_{u_1,u_2}$ is unique, and thus it must be shared by both $\Pi$ and $\Pi'$. The latter implies that every vertex in the subpath must have degree degree exactly $2$. Therefore, the only non-terminals in $\Pi \cup \Pi'$ are vertices $u_1$ and $u_2$, since non-terminals of degree two are removed from the edge contractions performed in the algorithm. 

There are $\mathcal{O} (f(k)^2)$ pairs of shortest paths from $\mathcal{P}'$, each having at most $2$ non-terminals. Hence, the number of non-terminals in $H$ is $\mathcal{O} (f(k)^2)$.
\end{pfof}

\subsection{Proof Sketch of Theorem \ref{thm:treewidth}}\label{app:treewidth}

Next, we present better guarantees for bounded treewidth graphs.
These improvements make crucial use of the fact that such graphs admit small separators:
given a graph $G$ of bounded treewidth $p$ and any nonnegative vertex weight function $w(\cdot)$, there exists a set $S \subset V(G)$ of at most $p+1$ vertices
whose removal separates the graph into two connected components, $G_1$ and $G_2$, each with $w(V(G_i)) \leq 2/3 w(V(G))$ (see \cite{bodlaender95}).

Krauthgamer et al.~\cite{distancepreserving} use the above fact to construct an $(1, \mathcal{O}(p^{3}k))$-DAM for graphs of treewidth at most $p$.
We show that with two modifications, their algorithm can be extended for the constructions of DAMs.
The first modification is Step 2 of the algorithm \textsc{ReduceGraphTW} in \cite{distancepreserving}.
For any integer $q\geq 1$, we replace their call to \textsc{ReduceGraphNaive}$(H,T \cup B)$\footnote{We remark that they use $R$ to denote the set of terminals.}
by our procedure \textsc{MinorSparsifier}$(H,T \cup B, \mathcal{P}')$, where $ \mathcal{P}'$ is a $(2q-1, \mathcal{O}(p^{1+1/q}))$-TPc of $G$.

The second modification is a generalization of Lemma 4.2 in \cite{distancepreserving}. The main idea is to use the small separator set $S$ to decompose the graph into smaller almost-disjoint graphs $G_1$ and $G_2$, compute their DAMs recursively, and then combine them using the separator $S$ into a DAM of $G$. This implies that the separator $S$ must belong to each $G_i$, i.e. all non-terminal vertices of $S$ must be counted as additional terminals in each $G_i$. Below we give a formal definition of this decomposition/composition process.

Let $G_1 = (V_1, E_1, \ell_1)$ and $G_2 = (V_2,E_2, \ell_2)$ be graphs on disjoint sets of non-terminals, having terminal sets $T_1 = \{s_1,s_2,\ldots,s_{a_1}\}$ and $T_2 = \{t_1,t_2,\ldots,t_{a_2}\}$, respectively. Further, let $\phi(s_i)=t_i$, for all $i=1,\ldots,c$ be an one-to-one correspondence between some subset of $T_1$ and $T_2$ (this correspondence is among the separator vertices). The \textit{$\phi$-merge} (or $2$-sum) of $G_1$ and $G_2$ is the graph $G = (V,E,\ell)$ with terminal set $T = T_1 \cup \{t_{c+1},\ldots,t_{a_2}\}$ formed by identifying the terminals $s_i$ and $t_i$, for all $i=1,\ldots,c$, where $\ell(e) = \min\{\ell_1(e),\ell_2(e)\}$ (assuming infinite length when $\ell_i(e)$ is undefined). We denote this operation by $ G := G_1 \oplus_{\phi} G_2$.

Below we state the main lemma whose proof goes along the lines of \cite[Lemma 4.2]{distancepreserving}.

\begin{lemma} Let $G = G_1 \oplus_{\phi} G_2$. For $j=\{1,2\}$, let $H_j$ be an $(\alpha_j,f(a_j))$-\emph{DAM} for $G_j$. Then the graph $H = H_1 \oplus_{\phi} H_2$ is an $(\max\{\alpha_1,\alpha_2\}, f(a_1) + f(a_2))$-\emph{DAM} of G.  
\end{lemma}

In ~\cite{distancepreserving} it is shown that the size of the minor returned by the algorithm \textsc{ReduceGraphTW} is bounded by the number of leaves the in the recursion tree of the algorithm. Further, they prove that there are at most $\calO(k/p)$ such leaves. Plugging our bounds from the modification of Step 2 along with the above lemma yields Theorem \ref{thm:treewidth}.

}

\end{document}